\documentclass[12pt]{iopart}

\usepackage[utf8]{inputenc}
\usepackage[T1]{fontenc}
\usepackage{color}

\usepackage{iopams}
\expandafter\let\csname equation*\endcsname\relax
\expandafter\let\csname endequation*\endcsname\relax
\usepackage{amsmath}
\usepackage{amstext}
\usepackage{amsthm}

\usepackage{graphicx}
\usepackage{braket}
\newtheorem{lemma}{Lemma}

\newtheorem{proposition}{Proposition}

\newcommand\pp{\mathfrak{p}}
\newcommand\nn{\mathfrak{n}}
\newcommand\PP{\mathfrak{P}}
\newcommand\NN{\mathfrak{N}}

\begin{document}
\title{Receiver operation characteristics of quantum state discrimination}
\author{A. Bodor}
\address{Department of Physics of Complex Systems, Faculty of Science, Lor\'and E\"otv\"os University, Budapest,H- 1117 Budapest, Pázmány Péter sétány 1/A, Hungary}
\author{M. Koniorczyk}
\address{Institute of Mathematics and Informatics, Faculty of Sciences, University of P\'ecs, H-7624 P\'ecs, Ifj\'us\'ag \'utja 6., Hungary} 
\ead{kmatyas@gamma.ttk.pte.hu}
\vspace{10pt}
\begin{indented}
\item[]March 2016
\end{indented}
\begin{abstract}
  We provide a description of the problem of the discrimination of two
  quantum states in terms of receiver operation characteristics
  analysis, a prevalent approach in classical statistics. Receiver
  operation characteristics diagrams provide an expressive
  representation of the problem, in which quantities such as the
  fidelity and the trace distance also appear explicitly. In addition
  we introduce an alternative quantum generalization of the classical
  Bhattacharyya coefficient. We evaluate our quantum Bhattacharyya
  coefficient for certain situations and describe some of its
  properties. These properties make it applicable as another possible
  quantifier of the similarity of quantum states.
\end{abstract}

\pacs{03.67.-a, 03.65.Aa}
%
%
\submitto{\JPA}
%
%
%

\section{Introduction}

The problem of discriminating two quantum states is a central,
well-understood problem of quantum information processing: given a
quantum system along with the prior information that it is in either
of two possible given states with some probability, we have to tell,
as accurately as possible, in which of its possible states the system
is. There has been a tremendous and successful effort, resulting in a
detailed understanding of many facets of the problem. Instead
of recapitulating these results here, we refer to the excellent
reviews in Refs.~\cite{QSD, 1751-8121-48-8-083001}.

Note that the discrimination problem has its antecedents in classical
statistics. Given a sample drawn from either of two possible random
variables with the same possible set of values but different
distribution, a typical task of statistics is to decide which of the
variables were actually used to draw the sample.  This relates to the
idea of receiver operation characteristics (ROC) analysis, which can
be simply understood from one of its typical applications. The task is
to decide whether a patient suffers from a given illness. In order to do
so, a test is carried out, from which a conclusion is drawn. Of
course, there may be four cases. If we conclude that the patient is
ill, this maybe a true or a false positive depending on whether we
succeed or fail. True and false negatives can be understood similarly
for the negative conclusion. The rate of these events depend of course
on the test we carry out and the rules to draw the conclusion. A given
definition of the test and the processing rules specify a
\emph{discriminator}. The idea of ROC analysis is to draw the true
positive rate as a function of the false positive rate for a given
discriminator. Visualizing all the possible discriminators in such a
way leads to a couple very intuitive and useful techniques termed as
ROC analysis (see Ref.~\cite{Fawcett:2006:IRA:1159473.1159475} for a
review).

The analogy with quantum state discrimination is clear: if we consider
one of the possible quantum states to be the ``positive'', the other
one to be the ``negative'' state of the system, we have to decide if
in reality it is in the ``positive'' state. Though in our case there
is no natural assymmetry in the choice of the ``positive'' or
``negative'' state, the technique is, as we shall see, applicable in
this case, too. The classical limit of this quantum scenario is the
case of the discrimination of two probability distributions, as
described above.

Throughout this paper we elaborate on the application of ROC analysis
to the quantum scenario. In order to do so, in
Section~\ref{sec:twobernoullis} we describe the case of the
discrimination of classical random variables, in terms of ROC
curves. Here we derive those facts of ROC analysis which we shall use
in the rest of the paper. It is of special interest here to understand
the classical Bhattacharyya-coefficient~\cite{MR0010358}, a quantity
characterizing the similarity of two probability distribution, as a
certain integral of the ROC curve. Though most of these are available
in the literature of ROC analysis, it is useful to describe them in
detail for sake of self-consistency. In Section~\ref{sec:ROCquantum}
we describe the case of ambiguous discrimination of two quantum
states. We provide a full analysis of the case of two pure states and
that of two mixed states with common support, and describe the most
general case of two arbitrary states, too. We find the representation
of trace distance and fidelity in the ROC curve. In addition, we
introduce and analyze a version of quantum Bhattacharyya coefficent as
an integral of a ROC curve, just like in the classical case. This approach is
different from the quantum Bhattacharyya coefficient that was introduced in
the literature~(see e.g. Ref.~\cite{9780511535048}). We explore
several properties of this quantity which make it a useful similarity
measure~\cite{1407.3739}. In Section~\ref{sec:unambiguous} we outline
the case of unambiguous state discrimination in the ROC picture. In
Section~\ref{sec:conclusions} our results are summarized and
conclusions are drawn.

\section{Classical ROC curves}
\label{sec:twobernoullis}

Let us consider the problem of discriminating two classical
random variables with the same finite range (that is, the
discrimination of two probability distributions). Though in ROC
analysis it is common to consider empirical rates, the technique works
also for probability distributions and conditional probabilities, too,
and this is the way we follow. First we consdider the case of two
binary variables.

\subsection{Two binary variables}

Assume that we have a binary random variable (i.e. a classical bit)
with 0 and 1 as possible values. The variable can be distributed
according to two possible distributions
\begin{equation}
 (p, 1-p)\quad \text{or}\quad (q, 1-q),  
\label{eq:distros}
\end{equation}
where $p,q \in[0,1]$. We know these distributions in advance. We want
to decide whether the variable is distributed according to the first
of these. We denote this event by $\PP $ (``positive''), while the
complementary event with $\NN $ (``negative''). Let $\Pr(\PP )=\lambda\in
[0,1]$ be also given in advance. We measure the variable, the
measurement yields the actual value. We denote the two possible events
with the respective values $0$ and $1$. The task is to guess which one
is the true distribution. We denote the events corresponding to our
two possible guesses by $\pp $ and $\nn $.

By chosing a particular classifier we fix the way how we deduce the
guess $\pp/\nn$ from the result of the observation $0/1$. The
properties of the possible classifiers will be studied in the ROC
space: a classifier is represented here by a point, on the vertical
axis we have the conditional probability $\Pr(\pp|\PP)$, while on the
horizontal axis we measure $\Pr(\pp|\NN)$. If we were to repeat the
procedure on a large number samples all prepared in the same way, the
first quantity describes the \emph{true positive rate}, while the
second one the \emph{false positive rate}.

For a classifier we consider a general random mapping characterized
by the following conditional probabilities:
\begin{eqnarray}
  \label{eq:classifier}
  \Pr(\pp|0) = p_{\text{a},\pp}, \nonumber \\
  \Pr(\nn|1) = p_{\text{a,}\nn}, \nonumber \\
  \Pr(\nn|0) = 1-p_{\text{a,}\pp}, \nonumber \\
  \Pr(\pp|1) = 1-p_{\text{a,}\nn},
\label{indicator}
\end{eqnarray}
with the \emph{acceptance probabilities} $p_{\text{a,}\pp},
p_{\text{a,}\nn}\in [0,1]$. Setting
$p_{\text{a,}\pp}=p_{\text{a,}\nn}=1$ corresponds to the deterministic
case when we consider $0$ as the indicator of $\PP$ and $1$ that of
$\NN$. Alternatively, we may allow for some randomness in our
choice. Note that the choice of the distribution ($\PP/\NN$), the
measurement result ($0/1$) and the conclusion ($\pp/\nn$), as random
variables, form a Markov-chain.

In the given situation we have
\begin{equation}
  \label{eq:tpderive}
  \Pr(\pp|\PP) = \Pr(\pp |0,\PP) \Pr(0|\PP) + \Pr(\pp |1,\PP) \Pr(1|\PP).
\end{equation}
We have $\Pr(\pp|0,\PP) = \Pr(\pp |0)$ and $\Pr(\pp|1,\PP) =
\Pr(\pp|1)$ from the Markov chain, and these probabilities are known
from Eq.~\eqref{indicator}, while the rest of the probabilities is
known from \eqref{eq:distros}, so we have
\begin{equation}
  \label{eq:tp}
  \Pr(\pp|\PP) = p_{\text{a,}\pp} p + (1-p_{\text{a,}\nn}) (1-p).
\end{equation}
For the false positive rate we have
\begin{equation}
  \label{eq:fpderive}
  \Pr(\pp|\NN) = \Pr(\pp |0,\NN) \Pr(0|\NN) + \Pr(\pp |1,\NN) \Pr(1|\NN),
\end{equation}
which evaluates along the same lines of thought as
\begin{equation}
  \label{eq:fp}
  \Pr(\pp|\NN) = p_{\text{a,}\pp} q + (1-p_{\text{a,}\nn}) (1-q).
\end{equation}
Thus for a given $p$ and $q$ we have to plot $p_{\text{a,}\pp} p + (1-p_{\text{a,}\nn}) (1-p)$
against $p_{\text{a,}\pp} q + (1-p_{\text{a,}\nn}) (1-q)$ for all possible $(p_{\text{a,}\pp}, p_{\text{a,}\nn})$ pairs
as parameters to obtain the ROC diagram for all the possible
classifiers. Note that it does not depend on $\lambda$: it
characterizes the classifier based on the two probability
distributions to be distinguished, and not the prior probability of
having one of them at hand.

To visualize the behavior of the point set in the ROC space defined by
Eqs.~\eqref{eq:fp} and~\eqref{eq:tp}, we express $(1-p_{\text{a,}\nn})$
from Eq.~\ref{eq:fp}:
\begin{equation}
  \label{eq:1minusrho}
  0\leq (1-p_{\text{a,}\nn}) = \frac{x-p_{\text{a,}\pp} q}{1-q}\leq 1,
\end{equation}
where we have indicated the bounds on
$(1-p_{\text{a,}\nn})$. Substituting this to~\eqref{eq:tp}, after some
algebra we have
\begin{equation}
  \label{eq:curve1}
  y=\frac{1-p}{1-q}x+\left( p-\frac{1-p}{1-q}q\right)p_{\text{a,}\pp} = Ax+Bp_{\text{a,}\pp}.
\end{equation}
(We exclude the trivial cases when $p,q= 0\ \mbox{or}\ 1$.)
For $p_{\text{a,}\pp}=0$ this defines a line through the origin of the ROC
space. For $p_{\text{a,}\pp}=1$ we get a parallel line, above or below the other
one, depending on the sign of $B$. The available region is within
these parallel lines. However, $p_{\text{a,}\pp}$ may have stronger bounds due to
the bounds in Eq.~\eqref{eq:1minusrho}. These are described by the two
parallel lines
\begin{equation}
  \label{eq:additionalboundaries1}
  y=\frac{p}{q}x,
\end{equation}
and
\begin{equation}
  \label{eq:additionalboundaries2}
 y= \frac{p}{q}x+(1-p)-\frac{p}{q}(1-q).
\end{equation}
The available region of the ROC space is thus a parallelogram, the
convex hull of the points $(0,0)$, $(1,1)$, $(q,p)$, and
$(1-q),(1-p)$. An example is depicted in
Fig.~\ref{fig:twobernoulliROC}.
\begin{figure}
  \centering
  \includegraphics{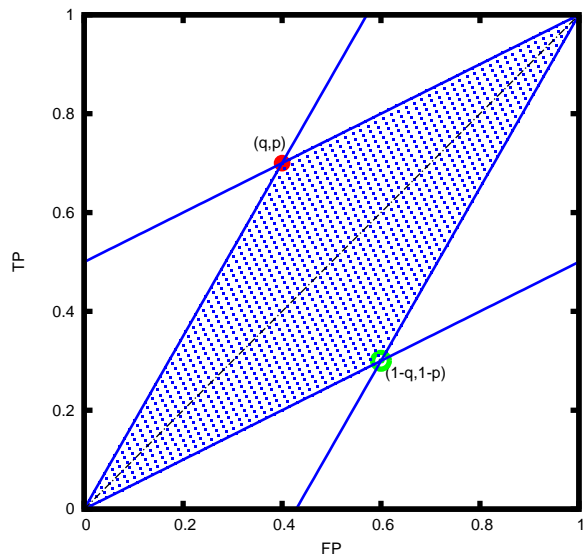}
  \caption{(color online) The available ROC region for the binary variables with
    $p=0.7$, $q=0.4$.}
  \label{fig:twobernoulliROC}
\end{figure}

In the ROC space, the better a discriminator performs, the closer its
representing point is to the point $(0,1)$ (representing a true
positive rate of 1, and a false positive rate of 0). Thus the
interesting part of the ROC diagram is the upper piecewise linear
curve $(0,0)\to \mathbf{P} \to (1,1)$, where the point $\mathbf{P}$ is
either $(q,p)$, or $(1-q, 1-p)$. In the latter case we can make the
substitutions $p\to 1-p$ and $q\to 1-q$, resulting in the 
diagram of the same shape, but the two possible $\mathbf{P}$ points
interchanged. We may thus always consider $\mathbf{P}=(q,p)$ to be on
the upper broken line.

Note that the line representing the best possible discriminators,
termed as the ``optimal ROC curve'' in what follows, connects the
points specified by the two possible cumulative distributions
$(0,q,1)$ and $(0,p,q)$ plotted against each other.

It is important to point out that the ROC domain is convex: when
considering two points representing two feasible discriminators, all
the discriminators whose points are on the line interconnecting the
two chosen discriminators on the ROC space are feasible, too. The
discriminators on interconnecting line represent discriminators which
are defined so that we randomly apply either of the two considered
discriminators, according to the evaluation of a random binary
variable~\cite{Fawcett:2006:IRA:1159473.1159475}.  Note that the
optimal choice of the parameters $p_{\text{a,}\pp}$ and
$p_{\text{a,}\nn}$ is not completely trivial. However, at the point
$(q,p)$ we have $p_{\text{a,}\pp}=p_{\text{a,}\nn}=0$, corresponding
to the case when we choose to conclude $\pp$ if $1$ is measured, and $\nn$
if $0$ is measured, deterministically. Of course, these leads to a TP
rate of $p$ and an FP rate of $q$. The points $(0,0)$ and $(1,1)$
correspond to a (though less useful) choice of $p_{\text{a,}\pp}=0,
p_{\text{a,}\nn}=1$ and $p_{\text{a,}\pp}=1, p_{\text{a,}\nn}=0$
respectively. As the available domain is convex, it is indeed
plausible that we have the broken line connecting these segments as
the ROC curve.

Let us describe the sets of points in the ROC space which correspond
to a constant overall probability of failure of the identification,
given a prior probability $\lambda$. This overall probability reads
\begin{eqnarray}
  \label{eq:failure}
  \Pr(\text{fail}) = \Pr(\pp,\NN) + \Pr(\nn,\PP) = \Pr(\NN) \Pr(\pp|\NN) + \Pr(\PP) \Pr(\nn|\PP) \nonumber \\
  = (1-\lambda) \Pr(FP) + \lambda (1-\Pr(TP)).
\end{eqnarray}
Thus for a given $\Pr(\text{fail})$, we have a line in the ROC space:
\begin{equation}
  \label{eq:constfail}
  \Pr(TP) = \frac{\lambda - \Pr(\text{fail})}{\lambda} + \frac{1-\lambda}{\lambda} \Pr(FP),
\end{equation}
with the slope of $(1-\lambda)/\lambda$.

Finally let us now consider the case of a random variable with $N$
possible values. It can be shown that with the appropriate choice of
the mapping of the measurement results (having $N$ possible values in
this case) to the conclusion ($\pp$ or $\nn$), we get the optimal ROC
curve by plotting the two cumulative distributions against each other,
and interconnecting the so obtained points with lines, in the
descending order of $p_i/q_i$.


\subsection{Bhattacharyya coefficient from the ROC curve}

Observe that the shape of the ROC curve is characteristic for the
distributions in question: if there are no points with the same
$p_i/q_i$ ratio, the curve just determines the two distributions. Only
points with the same $p_i/q_i$ introduce a possible ambiguity: they may
be interchanged, subdivided into additional points by introducing new
possible values of the variables with the same probability ratio, or
merged to correspond to the same value.

A prevalently used measure of the similarity of two distributions
can be geometrically deduced from the ROC curve.  Consider the
coordinate transformation
\begin{equation}
  \label{eq:transf1}
  \begin{pmatrix}
    q\cr p 
  \end{pmatrix}
\to 
\begin{pmatrix}
  t \cr s
\end{pmatrix}
=
\begin{pmatrix}
  \frac{p-q}{2} \\ \frac{p+q}{2}
\end{pmatrix},
\end{equation}
that is, we consider axes rotated by $\pi/4$ counterclockwise to the
original ROC axes, and the interchange of their ordering, and
shrinking with a factor of $\sqrt{2}$. Further, consider the following
Minkowski-metrics defined in terms of the new coordinates:
\begin{equation}
  \label{eq:minkowskimetrics}
  d{\ell}^2=ds^2-dt^2.
\end{equation}
In the terminology of special relativity the ``space'' axis will be
the original $(0,0)\to (1,1)$ diagonal of the ROC space; while the
``time'' axis is at right angles to this, pointing upwards.  In this
way, the part of the domain of the ROC space above the diagonal, where
the optimal ROC curve resides, has positive coordinates. Thus the
original TP axis is one of the edges of the ``future light cone'',
while the other axis is that of the ``past light cone''. Note that the
lines parallel to these axis have zero length in this metric.

Assume that the cumulative distributions
\begin{equation}
  \label{eq:cumulativeP}
  P_k=\sum_{l=1}^k p_k
\end{equation}
and
\begin{equation}
  \label{eq:cumulativeQ}
  Q_k=\sum_{l=1}^k q_k
\end{equation}
are plotted against each other, resulting in the ROC curve described
in the previous Section. The square of the length $\ell_k$ of the line
segment of the broken line (in Minkowski metric) between the points
$(Q_k,P_k)$ and $(Q_{k+1}, P_{k+1})$ reads
\begin{eqnarray}
  \label{eq:segmentlength}
  \ell^2_k = -\left( 
\frac{P_{k+1}-Q_{k+1}}2 - \frac{P_{k}-Q_{k}}{2} 
\right)^2
+\left(\frac{P_{k+1}+Q_{k+1}}2 - \frac{P_{k}+Q_{k}}{2} 
\right)^2 \\
=  -\left(\frac{p_k-q_k}2\right)^2 + \left(\frac{p_k+q_k}{2}\right)^2
= p_kq_k.
\end{eqnarray}
Hence, for the whole length of the broken line in Minkowski metric we
have
\begin{equation}
  \label{eq:bhattaclass}
  \ell =  \sum_k \ell_k = \sum_k \sqrt{p_kq_k} = B(p,q)
\end{equation}
is the Bhattacharyya-coefficient~\cite{MR0010358} prevalently used in
the literature of statistics. It measures in a way the similarity
between two probability distributions.

\section{The ROC curves for quantum states}
\label{sec:ROCquantum}

In case of ambiguous quantum state discrimination, we are
given a system in an unknown quantum state, but we know in advance
that the state is $\varrho_{\PP}$ with probability $\lambda$ and
$\varrho_{\NN}$ with probability $1-\lambda$. As we discriminate two
states, we can assume that we work in a two-dimensional subspace of
the Hilbert-space of the studied system, hence, $\varrho_{\PP}$ and
$\varrho_{\NN}$ are $2\times 2$ density matrices. They are known a priori,
along with $\lambda \in[0,1]$.

In case of ambiguous state discrimination we have one sample, we are
allowed to perform any generalized (POVM) measurement, and we have to
decide which one of the two states was given, with the lowest
probability of error. The minimum probability of error is given by the
well-known Helstrom-formula~\cite{Helstrom}:
\begin{equation}
  \label{eq:Helstrom}
  P_{\text{E, min}}=\frac12\left(1-|| \lambda \varrho_{\PP} -(1-\lambda)\varrho_{\NN}||_1\right).
\end{equation}
In case of ROC analysis, we consider the two possible conclusions $\pp$
(The state was $\varrho_{\PP}$ and $\nn$ (the state was $\varrho_{\NN}$). As for
the measurement, we consider the POVM leading to the optimum in the
Helstrom formula. This is formed by the spectral projectors associated
to the positive and negative eigenvalues of the Hermitian matrix
\begin{equation}
  \label{eq:LambdaMatrix}
  \Lambda= \lambda \varrho_{\PP} -(1-\lambda)\varrho_{\NN}.
\end{equation}

\subsection{Two pure states}
\label{sec:twopure}

In this Section we consider the discrimination of two pure
states. This problem is two-dimensional by nature, so the
discrimination of qubit states is the most general as long as mixed
states are not considered.

We are given a quantum bit in either of the states
\begin{equation}
  \label{eq:twopure}
  \ket{\Psi_{\PP}} =
  \begin{pmatrix}
    \cos{\frac{\theta_p}{2}}\cr
    \sin{\frac{\theta_p}{2}}
  \end{pmatrix}
\quad \text{or}\quad
  \ket{\Psi_{\NN}} =
  \begin{pmatrix}
    \cos{\frac{\theta_q}{2}}\cr
    \sin{\frac{\theta_q}{2}}
  \end{pmatrix}
\end{equation}
as the ``positive'' and ``negative'' case.
(To compare with the classical case, set $p=\cos^2{(\theta_p/2)}$ and
$q=\cos^2{(\theta_q/2)}$ in Eq.~\eqref{eq:distros}). For symmetry reasons we
may consider real vectors, but we allow the azimuthal angles
$\theta_p$ and $\theta_q$ to be in the interval $[0,2\pi[$, thereby
obtaining two arbitrary pure states in the $x-z$ plane of the
Bloch-sphere.

As the measurement we consider a projective measurement in the basis
  \begin{equation}
  \label{eq:measurementbasis_pure}
  \ket{\Phi_{\pp}} =
  \begin{pmatrix}
    \cos{\frac{\alpha}{2}}\cr
    \sin{\frac{\alpha}{2}}
  \end{pmatrix},\quad
  \ket{\Phi_{\nn}} =
  \begin{pmatrix}
    -\sin{\frac{\alpha}{2}}\cr
    \cos{\frac{\alpha}{2}}
  \end{pmatrix}, \quad \alpha\in[0,2\pi[.
\end{equation}
We assume now that the measurement result corresponding to
$\ket{\Phi_{\pp}}$ leads to the conclusion that we were given
$\ket{\Psi_{\PP}}$, and the same holds for the case of
``false''. We shall discuss the effect of the classical postprocessing
as described in the previous section later.

To obtain the ROC representation we need the false positive and the
true positive probabilities, which read, using the same reasoning as in
the classical case, after a straightforward calculation,
 \begin{eqnarray}
   \label{eq:tpfpquantum}
   \Pr(TP)= |\braket{\Phi_{\pp}|\Psi_{\PP}}|^2=
   \frac{1+\cos(\alpha -  \theta_p)}{2} 
\nonumber \\
   \Pr(FP)=|\braket{\Phi_{\nn}|\Psi_N}|^2=\frac{1+\cos(\alpha -
    \theta_q)}{2},
 \end{eqnarray}
 again irrespectively of the prior probability of being given either
 of the possible states. This is apparently the parametric equation of
 an ellipse. This is visualized in Fig.~\ref{fig:roctwopure}.
\begin{figure}
  \centering
  \includegraphics{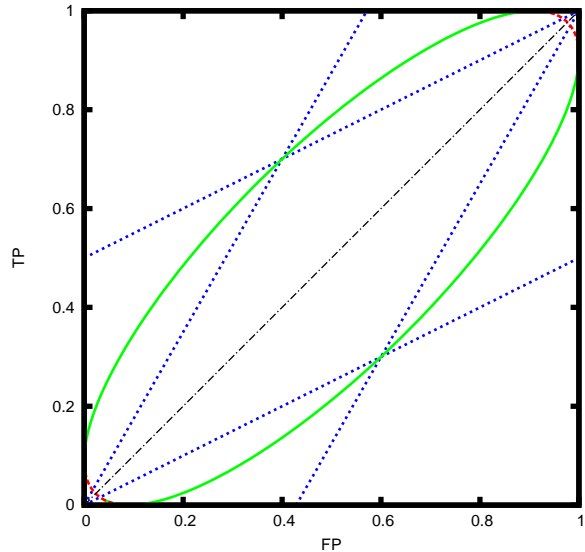}
  \caption{(color online) The ROC curve for two pure quantum states corresponding to
    the classical binary variables $p=0.7$, $q=0.4$ is the ellipse
    in the figure. For convenience we have also plotted the lines
    relevant to the classical case.  The parts in dashed red are not
    accessible by using the measurement corresponding to optimal
    ambiguous discrimination of the states (from the Helstrom
    formula).}
  \label{fig:roctwopure}
\end{figure}
Our ellipse's main axis lies on the line $(0,0)\to (1,1)$.  The edge
points $(q,p)$ and $(1-q, 1-p)$ are located on the ellipse. This just means that each of the equations
\begin{eqnarray}
   \label{eq:tpfpquantumclass}
   \cos\left( \frac{\theta_p}{2}\right)^2 =
   \frac{1+\cos(\alpha -  \theta_p)}{2},\ \text{and}
\nonumber \\
   \cos\left( \frac{\theta_q}{2} \right)^2
 =\frac{1+\cos(\alpha - 
    \theta_q)}{2}
 \end{eqnarray}
 indeed do have a solution for $\alpha$.

Let us now describe the points at which the ellipse touches the edges
of the ROC-square. The true positive rate equals to one when $\alpha
 - \theta_p = 0$, that is, $\ket{\Phi_p}=\ket{\Psi_{\PP}}$. The false
 positive rate at this point is thus just the fidelity of the two
 states to be distinguished:
 \begin{equation}
   \label{eq:purefidel}
   \mathcal{F}(\Psi_{\NN}, \Psi_{\PP})=|\braket{\Psi_{\PP}|\Psi_{\NN}}|^2=
|\braket{\Phi_{\pp}|\Psi_{\NN}}|^2.
 \end{equation}
 The false positive rate is one at $\alpha - \theta_q=0$, thus
 $\ket{\Phi_p}=\ket{\Psi_N}$, yielding the fidelity again ,as
 \begin{equation}
   \label{eq:purefidel2}
   \mathcal{F}(\Psi_{\NN}, \Psi_{\PP})=|\braket{\Psi_{\PP}|\Psi_{\NN}}|^2=
|\braket{\Psi_{\PP}|\Phi_{\NN}}|^2.
 \end{equation}
 The true positive rate is zero when $\alpha - \theta_p = \pi$, that
 is, $\alpha=\pi+\theta_p$ hence $\ket{\Phi_p}=\ket{\Psi_{\PP}^\bot}$,
 where $\ket{\Psi_{\PP}^\bot}$ is a state orthogonal to $\ket{\Psi_{\PP}}$ in
 the Hilbert-space. therefore we have for the false positive rate
 \begin{equation}
   \label{eq:purefidel3}
   |\braket{\Phi_{\pp}|\Psi_{\NN}}|^2=|\braket{\Psi_{\PP}^\bot|\Psi_{\NN}}|^2=1-\mathcal{F}(\Psi_{\NN}, \Psi_{\PP}).
 \end{equation}
 The point for which the false positive rate is zero is at
 $\alpha=\theta_n+\pi$, with the true positive rate of
 $1-\mathcal{F}(\Psi_{\NN}, \Psi_{\PP})$, which can be justified along the
 same lines.  Hence, we can determine the fidelity of the two
 states immediately according to the ROC curve.

The classical postprocessing by defining probabilities to accept or
reject true or false measurement results in points within the
ellipse, making the full area of the ellipse accessible by some
setup. Note, that this is ambiguous in the sense that different
choices of measurement and postprocessing may result in the same point
in the ROC plane.

One may follow a different approach: given the
states to be discriminated, and the prior probability $\lambda$, we
may determine the spectral projectors from
Eq.~\eqref{eq:LambdaMatrix}, and consider the resulting projective
measurements. In this case instead of the direct parameter $\alpha$ of
the measurement, we have $\lambda$ as a parameter. Plotting the
resulting points in the ROC curve, we obtain a part of the upper and
lower parts of the ellipse, as depicted in
Fig.~\ref{fig:roctwopure}. The missing parts of the ellipse (in cyan
in the figure) are accessible by the spectral projectors of the
overall prior density matrix
\begin{equation}
  \label{eq:RhoMatrix}
  \rho= \lambda \varrho_{\PP}  + (1-\lambda)\varrho_{\NN}.
\end{equation}
Thus we have two different interpretations: either given the two
states, we may either consider all the possible measurements
parametrized by $\alpha$ to obtain the ellipse, or we may consider all
the possible prior probabilities $\lambda$, and calculate the optimal
measurement, whose corresponding point shall give the upper and lower
part of the measurement.

To find the prior probability $\lambda$ for which a point on the curve
in the ROC space is optimal, recall that the curves with a constant
overall failure probability are the lines in
Eq.~\eqref{eq:constfail}. Assume that we are in a point at the upper
part of the ellipse (i.e. we fix a measurement direction
$\alpha$). Recall also that a point is optimal if it is as close to
the point (1,0) as possible. In order to exclude the possibility of
decreasing this distance, provided that we fix a prior probability
(and thus the slope of the line), the line must be a tangent. 
The slope of the tangent of the ellipse in the given point is, on the other hand,
\begin{equation}
  \label{eq:tangent}
  \frac{\lambda}{1-\lambda} = \frac{d \Pr(TP)}{d\alpha} \left(\frac{d \Pr(FP)}{d\alpha} \right)^{-1} =
  \frac{\sin(\theta_p-\alpha)}{\sin(\theta_q-\alpha)}.
\end{equation}
The relation linking the parameter $\alpha$ in our first approach, to
the corresponding parameter $\lambda$ in the second can be found from
this equation. In addition, we have an operative notion for the
tangents of the ROC curve, which is in fact derives from completely
general features of the ROC curve.

It is of some interest to evaluate the classical Bhattacharyya
coefficient for the distributions $\Pr(TP), 1-\Pr(TP)$ and $\Pr(FP),
\Pr(1-FP)$, by substituting these probabilities from
Eq.~\eqref{eq:tpfpquantum} into~\eqref{eq:bhattaclass}. For the point
corresponding to the classical case, it is just the classical
coefficient found in Eq.~\eqref{eq:bhattaclass}. For the upper part of
the ellipse, we find that the classical Bhattacharyya coefficient is
constant for all the points. Thus the upper ellipse segment
corresponding to the different discriminators for two pure qubits (or,
otherwise speaking, the optimal discriminators for different prior
probabilities $\lambda$) are on the isolines of the classical
Bhattacharyya coefficient.

\subsection{A quantum Bhattacharyya coefficient} 
\label{sec:qbatta_pure}

According to our experience with the length in Minkowski metric of the
upper part of the ROC curve, the question naturally arises whether the
same integral for the upper part of the ellipse, from the points where
it touches the $0-1$ and $1-1$ axes is informative for the quantum
case. (Note that the optimal ROC curve contains the vertical line
connecting the $(FP=0,TP=0)$ point to the point of tangency of the
ellipse, as well as the horizontal segment from the other point of
tangency to the $(FP=1,TP=1)$ point, but these are of zero
Minkowski-lenth.

To calculate this we carry out the same transformation as in
Eq.~\eqref{eq:transf2}.
\begin{equation}
  \label{eq:transf2}
  \begin{pmatrix}
    \Pr(FP)\cr \Pr(TP) 
  \end{pmatrix}
\to 
\begin{pmatrix}
  t \cr s
\end{pmatrix}
=
\begin{pmatrix}
  \frac{\Pr(TP) -\Pr(FP)}{2} \\ \frac{\Pr(TP) +\Pr(FP)}{2}
\end{pmatrix},
\end{equation}
and substitute the results from Eq.~\eqref{eq:tpfpquantum} into
this. We obtain
\begin{eqnarray}
  \label{eq:xymink}
  t= \frac14\left( 
    \left( \cos(\theta_p) - \cos (\theta_q)  \right) \cos( \alpha) +
    \left( \sin \theta_p) - \sin (\theta_q)  \right) \sin( \alpha) 
\right)\nonumber \\
  s= \frac14\left( 
    \left( \cos(\theta_p) + \cos (\theta_q)  \right) \cos( \alpha) +
    \left( \sin \theta_p) + \sin (\theta_q)  \right) \sin( \alpha) 
\right).
\end{eqnarray}
 The length of the curve in Minkowski metric is obtained then as
\begin{equation}
  \label{eq:Bquantum}
  B(\ket{\Psi_{\PP}}, \ket{\Psi_{\NN}}) = \int\limits_{\text{upper curve}} \sqrt{ \left(\frac{ds}{d\alpha} \right)^2 -
\left(\frac{dt}{d\alpha} \right)^2}  d\alpha.
\end{equation}
As we integrate along the upper part of the curve, we go through
alphas from $\Pr(FP)=0$ to $\Pr(TP)=1$. In the particular case,
however, we have to choose the integration domain to go through the
upper curve in the proper direction.

After some algebra we get
\begin{equation}
  \label{eq:integral}
B(\ket{\Psi_{\PP}}, \ket{\Psi_{\NN}}) =
\frac{\sqrt{2}}{4}
\int\limits_{\text{upper curve}} \sqrt{
\cos(\theta_q+\theta_p)-
\cos(\theta_q+\theta_p-2\alpha)
}
d\alpha.  
\end{equation}
This leads to elliptic integrals, and can be evaluated numerically.
Let us now consider the particular choice of $\theta_p=0$, and
$\theta_q\in [0,\pi]$, thus the positive state points upwards on the
Bloch-sphere, while the other one ranging from this state to the one
pointing downwards, being orthogonal to the other in the
Hilbert-space. This choice covers all the relevant cases.  
In this case the integral in Eq.~\eqref{eq:integral} reads
\begin{equation}
  \label{eq:integralp0}
B(\ket{\Psi_{\PP}}, \ket{\Psi_{\NN}}) = 
\frac{\sqrt{2}}{4}
  \int\limits_{\theta_q+\pi}^{2\pi} \sqrt{\cos(\theta_q)-\cos(2\alpha-\theta_q)}
d\alpha.
\end{equation}
This in fact can be expressed in a closed form as 
\begin{eqnarray}
  \label{eq:integralp1}
B(\ket{\Psi_{\PP}}, \ket{\Psi_{\NN}}) =
\frac{1}{4}
\left[
\frac{2 E(x|k)-(1-\cos\theta) F(x|k)}
{\sin\left(\frac\theta2-\alpha\right)}
\right]_{\theta_q+\pi}^{2\pi},
\end{eqnarray}
where 
\begin{equation}
  \label{eq:Ellx}
  x = \cos\left( \frac\theta2-\alpha\right) \sqrt{1+\cos\theta}, 
\end{equation}
and
\begin{equation}
  \label{eq:ellk}
  k = \frac{\sqrt{2}}{2} \sqrt{1+\cos\theta},
\end{equation}
and $F(x|k)$ and $E(x|k)$ are the incomplite elliptic integrals of the
first and second kind, respectively.

The dependence of $B$ on $\theta_q$ is plotted in
Fig.~\ref{fig:BFtwopure}. 
\begin{figure}
  \centering
  \includegraphics{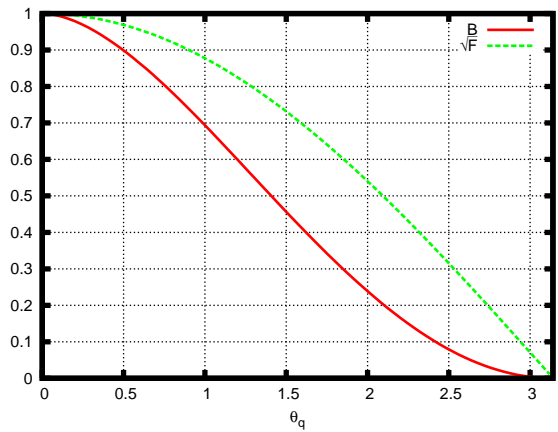}
  \caption{(color online) The quantum Bhattacharyya coefficient in
    Eq.~\eqref{eq:Bquantum}, and the square root fidelity of the two states for
    $\theta_p=0$, as a function of $\theta_q$.}
  \label{fig:BFtwopure}
\end{figure}
We have plotted the fidelity of the two states, too.  We can observe
that the fidelity is greater or equal than the Bhattacharyya
coefficient (equal for the values 0 and 1). Both are monotonously
decreasing.

\subsection{Two qubits}

Now we generalize our ideas to the general two-dimensional case, that
is, we allow the two states to be mixed, described by the
two-dimensional density operators $\varrho_{\PP}$ and $\varrho_n$. On the
Bloch-sphere, each of them can be written as $\varrho = 1/2 (\hat 1 +
{\mathbf{r}} {\mathbf{\sigma}})$, where ${\mathbf{r}}$ is the real
3-vector representing the state, while ${\mathbf{\sigma}}$ contains
the three Pauli matrices. For pure states, the $\mathbf{r}$-s are unit
vectors. The coordinates of the ROC curve in general read
 \begin{eqnarray}
   \label{eq:tpfpquantum_rho}
   \Pr(TP)= \tr(M_{\pp} \varrho_{\PP})
\nonumber \\
   \Pr(FP)= \tr(M_{\pp} \varrho_{\NN}),
 \end{eqnarray}
 where $M_{\pp}$ is the measurement operator for the positive
 conclusion. As trace is linear, and thus these expressions are
 inhomogeneous linear functions of the Bloch-vectors ${\mathbf{r}}_P$
 and ${\mathbf{r}}_N$, the ellipse we found in the case of two pure
 states is just shrinked anisotropically to the point
 $(1/2,1/2)$. (The TP and FP directions shrink proportionally with the
 length of the respective Bloch-vectors.) For mixed states, this does
 not touch the edges of the ROC space anymore. However, as the $(0,0)\to
 (1,1)$ line is necessarily part of the feasible ROC domain , the
 upper part of the ROC curve we are looking for is the convex hull of
 the (shrinked) ellipse and the points $(0,0)$, $(1-1)$.

 Recall that the ellipses found in \ref{sec:twopure} for two pure
 states are the iso-Bhattacharyya lines, and they touch the edges of
 the ROC space at points described by the fidelity of the two
 states. In case of mixed states, the square root fidelity is the
 minimum of the classical square root fidelities, that is, the
 Bhattacharyya coefficients of the classical probability distributions
 the two quantum states can yield when
 measured~\cite{9780511535048,PhysRevLett.76.2818}. Hence, to find the
 square root quantum fidelity we need to find the iso-Bhattacharyya
 ellipse which is tangent to the shrinked ellipse. The tangent point
 will give the optimal discriminator. The square root fidelity of the
 two mixed states is determined by the points where the tangent
 iso-Bhattacharyya ellipse touches the edges of the ROC space. Note
 that this consideration is valid for two dimensional states (that is,
 the two density operators' common support is two dimensional). For
 higher dimensional classical distributions there are no
 iso-Bhattacharyya curves, the existence of which is due to the fact
 that a pair of distributions is characterized by a single point in
 the ROC space.

\subsection{Two arbitrary quantum states}

Let us now omit any restriction, and consider $\varrho_{\PP}$ and
$\varrho_{\NN}$ in an arbitrary Hilbert-space. In this case we also need
to obtain a convex domain in the ROC space, whose upper edge is a ROC
curve corresponding to potentially optimal discrimination.  

This upper edge can be obtained by calculating the convex hull of the
points corresponding to different projective measurements resulting
from the Helstrom formula, for all possible prior probabilities
$\lambda$. The projective measurements $M_{\pp}, M_{\nn}$ resulting from the
Helstrom formula are projectors onto the positive vs. negative
spectral subspace of the matrix $\Lambda$ in
Eq.~\eqref{eq:LambdaMatrix}. Hence, the parts of the curve resulting
from the Helstrom formula directly have discontinuities: when the rank
of $M_{\pp}$ changes (due to the change of sign of an eigenvalue of
$\Lambda$), there is a split.

Alternatively we may just pick the two states and calculate all the
probabilities for all possible pairs of orthogonal projectors. If the
density operators' common support is $n$ dimensional, the examined
projectors' rank should range from $0$ to $n$. Each rank will result
in a convex spot around the line $(0,0),(1,1)$ in the ROC space (for
dimensions $0$ and $n$, we get the points $(0,0)$ and $(1,1)$
respectively. The convex hull of these will describe the accessible
domain in the ROC space, and its upper part will be our ROC
curve. This is illustrated in Fig.~\ref{fig:qroc_gen}.
\begin{figure}
  \centering
  \includegraphics{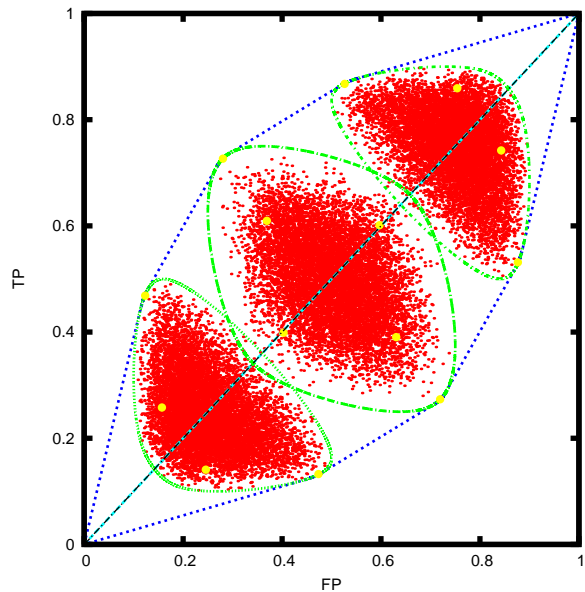}
  \caption{(color online) ROC curve/space of to randomly selected 4
    dimensional density matices. The dash-dotted (green) closed curves
    are regions accessible with rank 1, 2 and 3 projectors,
    respectively.  The dotted (blue) curve is the actual ROC curve,
    according to the Helstrom formula. The scattered small (red) dots
    are randomly (according to the Haar measure) selected
    projectors. The bigger, lighter (yellow) points are the those
    corresponding to the optimal fidelity measuring observables
    defined by Eq.~\eqref{eq:fidelitymeasurable}.}
  \label{fig:qroc_gen}
\end{figure}

Note that as the constant failure probability curves are the lines in
Eq.~\eqref{eq:constfail} (this is a general property of the ROC
space), the tangents of the ROC curve have the same operational
meaning as in the case of pure states: the tangent points give the
optimal discriminator for the given prior.

\subsubsection{Trace distance}

The ROC curve also yields a geometric notion of the trace distance.
Assume that we have equal prior probabilities, i.e.,
$\lambda=1/2$. Recall that according to Eq.~\eqref{eq:tangent}, the
optimal discriminator at the ROC curve is its point of tangency with
its $45$ degree tangent, i.e. the one parallel to the $(0,0)-(1,1)$
line. On the other hand, this discriminator is obtained from the
Helstrom formula, as the projective measurement $M_{\pp}, M_n$ onto the
positive vs. negative spectral subspace of the matrix $\Lambda$ in
Eq.~\eqref{eq:LambdaMatrix}, which in the present case reads
\begin{equation}
  \label{eq:LambdaMatrix12}
  \Lambda=\frac12 \left( \varrho_{\PP} -\varrho_{\NN}\right).
\end{equation}
The trace distance is defined as~\cite{9780511535048}
\begin{equation}
  \label{eq:tracedistdef}
  d_{\text{tr}}(\varrho_{\PP}, \varrho_{\NN}) = \frac12 \tr | \varrho_{\PP} - \varrho_{\NN} | = \tr |\Lambda|.
\end{equation}
(Note that we use the definition having the factor $1/2$ included.)
It can be written as
\begin{equation}
  \label{eq:tracedistderive}
  2d_{\text{tr}}(\varrho_{\PP}, \varrho_{\NN}) = \tr\left( M_{\pp} (\varrho_{\PP} - \varrho_{\NN}) \right) + \tr\left( M_{\nn} (\varrho_{\PP} - \varrho_{\NN}) \right).
\end{equation}
As $\Lambda$ is traceless, the two summands on the right hand side are equal, thus we have
\begin{equation}
  \label{eq:tracedistderive2}
  d_{\text{tr}}(\varrho_{\PP}, \varrho_{\NN}) = \tr\left( M_{\pp} (\varrho_{\PP} - \varrho_{\NN}) \right)=\Pr(TP)-\Pr(FP).
\end{equation}
This is just the intersection of the $45$ degree tangent with the TP
axis. The trace distance can be thus read directly from the ROC curve.

\subsection{The quantum Bhattacharyya coefficient} 

Following the line of thought of the previous sections, we finally
analyze the possibility of introducing a quantum Bhattacharyya
coefficient as we did for two pure states in
Section~\ref{sec:qbatta_pure}. Apart from the particular shape of the
ROC curve, we have not exploited any property stemming from the
quantum scenario. Hence, we may call the length of the ROC curve in
Minkowski metric a quantum Bhattacharyya coefficient. This can be
considered as a novel quantity to measure the similarity of the
states. It can be evaluated numerically, in a straightforward but
somewhat laborious way. We consider here its general properties here
instead.

\begin{proposition}
The quantum Bhattacharyya coefficient is zero if and only if the two
states can be distinguished with certainty (i.e., the two density
operators have disjoint support), and one if and only if the two
states are the same.  
\end{proposition}
\begin{proof}
In the first case, the ROC curve is the
$(0-0)\to(0,1)\to(1,1)$ broken line, which is the only one with zero length
in this metric, while in the latter case it is the $(0,0)\to(1,1)$
line, which is the longest line in the ROC space in this metric.  
\end{proof}

\begin{lemma}
\label{minklemma}
  Given two convex curves, the outer one (that is, which runs above
  the other in the ROC space) is shorter in the Minkowski metics.
\end{lemma}
\begin{proof}
  Consider the midpoint of the inner (lower) curve, and draw a tangent
  to it. This tangent intersects the outer (upper) curve at two
  points. Let us replace the outer curve by the line segment between
  theese two points. Thereby we have increased its length in the
  Minkowski metric. (If the two points are on both of the original
  curves, we skip this step.) Let us repeat this procedure with the
  points at $1/4$ and $3/4$ of the outer curve. By the replacements we
  increase the length of the outer curve. Upon repeating the procedure
  we get an arbitrarily fine approximation of the inner curve, and
  their Minkowski length shall be equal. As during the procedure we
  have not decreased the Minkowski length, the outer curve should have
  been shorter than the inner one. (Note that we have not even
  exploited the fact that the outer curve was convex.)
\end{proof}

\begin{proposition}
  The quantum Bhattacharyya coefficient is lower than or equal to the
  square root fidelity. It is equal if and only if the two density
  operators commute.
\end{proposition}
\begin{proof}
  It is known that there exists an optimal projective measurement, for
  which the Bhattacharyya coefficient of the classical measurement
  result is the square root of the fidelity. (In particular, the
  observable is
\begin{equation}
  \label{eq:fidelitymeasurable}
  M = \frac{1}{\sqrt{\varrho_{\NN}}} \sqrt{\sqrt{\varrho_{\NN}} \varrho_{\PP} \sqrt{\varrho_{\NN}}} \frac{1}{\sqrt{\varrho_{\NN}}} ,
\end{equation}
c.f Eq. (13.54) of Ref.~\cite{9780511535048}. This obviously defines
$d+1$ points in the feasible ROC domain, where $d$ stands for the
dimension of the joint support of the two density operators. (It is an
interesting open question whether these points are actually on the
optimal ROC curve.) These points define a sequence of line segments in
the feasible ROC domain, starting from $(0,0)$, ending at
$(1,1)$. It is obviously convex, and goes below the real ROC
curve. If the two density operator commute, then the ROC curve
coincides with this polyline, as the measurement itself becomes
classical. Otherwise the convex curve necessarily goes above its
chords. The statement then follows from Lemma~\ref{minklemma}.
\end{proof}

\begin{proposition}
  The quantum Bhattacharyya coefficient is monotone under
  completely positive maps.
\end{proposition}
\begin{proof}
  In order to see this we need to consider the whole domain in the ROC
  space which is accessible via a POVM in the ROC space (whose upper
  border is the ROC curve). Importantly this is symmetric to the point
  $(1/2,1/2)$: for an arbitrary point in the domain, the classical
  negation of the measurement result yields its reflected
  counterpart. In the Heisenberg picture, a completely positive (CP)
  map transforms a POVM into another POVM. Hence, such a
  transformation cannot lead out of the original accessible domain,
  but some points may become inaccessible. Thus the domain will shrink
  (or remain the same), and it will be of course still convex, as it
  is the ROC region of the new states. Thus the transformed ROC curve
  will be the same or below the upper curve, which has a higher length
  in the Minkowski metric according to Lemma~\ref{minklemma}
\end{proof}

We remark here that the parts of the ROC domain belonging to a given
rank of the projectors themselves are not contracted as described
above, in general. We conjecture that for unital maps, however, they
will be contracted, too.

\section{Unambiguous quantum state discrimination}
\label{sec:unambiguous}

Using the ROC technique we have introduced thus far can be fruitful in
the understanding of unambiguous quantum state discrimination,
too. Consider two arbitrary quantum states, the ``positive'' and
``negative'' one. (They do not need to be pure, neither of common
support.) We are looking for a POVM of 3 operators: $M_{\pp}$ detecting
the positive state with certainty, $M_{\nn}$ the negative one with
certainty, and the third one, $M_?=\hat 1 - M_{\pp} -M_{\nn}$ resulting in the
inconclusive outcome.

The first idea is to consider to merge one of the conclusive outcomes
with the inconclusive one. If the negative one is merged, this results
in an ambiguous discrimination problem, with a TP rate of one,
however. So it has a point in the ROC space, and it should be on the
line $(0,1)\to (1,1)$. If the positve one is merged, the negative
state is identified with certainty, thus we are on the TP axis. This
confirms very transparently the otherwise known fact that it is a
necessary condition for the availability of an unambiguous
discrimination that the ROC domain of the states touches both the
TP axis and the line $(0,1)\to (1,1)$. (Excluding the trivial points
$(0,0)$ and $(1,1)$ of course.)

Next we show that it is a sufficient condition as well. Consider a
pair of states which satisfy this condition. For each of the points
mentoned in the condition, we have a projective measurement, with the
property that one of its elements identifies one of the
states with certainty, as they project onto the other density
operator's null space. The Helstrom-formula provides us with these
projectors. (In fact, we may also use a positive measurement operator
that maps into the respective null-space.) Let us denote these
projectors by $\tilde M_{\pp}$ and $\tilde M_{\nn}$. These projectors can be
derived from the Helstrom formula, too.  Note that they need not be rank
one projectors in general. Their rank is between 1 and $d-1$, $d$
standing for the dimension of the common support of the states, for if
any of them has full rank, the respective measurement will not yield
any information. Hence, if the two density operators have fully
overlapping supports, the states cannot be unambiguously discriminated.

If we do have nontrivial $\tilde M_{\pp}$ and $\tilde M_{\nn}$, there
exist $\lambda_1, \lambda_2 \in ]0,1]$ so that $\lambda_1 \tilde
M_{\pp} + \lambda_2 \tilde M_{\nn} \leq \hat 1$
E.g. $\lambda_1=\lambda_2=1/2$ is a suitable choice. The choice
of the $\lambda$-s can tune the success probability of the detection
in favor one of the states and to the detrement of the other. Choosing
$\tilde M_? = \hat 1-\tilde M_{\pp} -M_{\nn}$ completes the two
operators the triplet $\tilde M_{\pp}, \tilde M_{\nn}, \tilde M_?$
which is a POVM that implements the unambiguous discrimination, not
necessarily optimally though.

\section{Conclusions}
\label{sec:conclusions}

We have presented the ROC analysis of the task of distinguishing two
quantum states, providing an intuitive representation of the
problem. The trace distance of two arbitrary states, and the fidelity
of a pair of states with a two-dimensional common support can be
direcly read out from the ROC diagram. We have also discussed to some
extent the case of unambiguous discrimination in this picture.

We have introduced a quantum generalization of the classical
Bhattacharyya coefficient as the length in Minkowski metric of the
optimal ROC curve for two arbitrary states. While it is a natural
generalization of the classical Bhattacharyya coefficient as evaluated
using the ROC curve, it has a notion which is different from the
earlier definitions of the quantum Bhattacharyya coefficient. We have
shown that this quantity is a lower bound for the fidelity, with
equality if and only if the states commute. As it is zero for
distinguishable states whereas it is one if the two states are the
same, it can be used to measure indistinguishability of quantum
states. In addition, it is monotone under completely positive maps.

Even though the facts covered by our present analysis are not novel,
they appear here in a comprehensive picture which supports a useful
way of thinking about the problem. We belive that this deepens the
understanding of the nature of the problem of distinguishing quantum
states. In addition, the present approach may prove to be useful in
understanding other strategies of quantum state discrimination as well
as more complex topics such as the discrimination of quantum channels,
etc.

The alternative quantum Bhattacharyya coefficient which we have
introduced is a natural generalization of the classical one, and even
though it is not straightforward to calculate it for particular
states, its notion and behavior is very transparent from the knowledge
of the ROC curve.  In addition, we found that it has properties which
are commonly expected from a well-behaved quantum distinguishability
measure.

\section*{Acknowledgements}
We thank Lajos Di\'osi for useful discussions and P\'eter Gn\"adig for
his hint to proove Lemma~\ref{minklemma}. A.~B. thanks Tam\'as Geszti
and Istv\'an Csabai for their support.

\section*{References}
\providecommand{\newblock}{}

\end{document}